\documentclass[submission,copyright,creativecommons]{eptcs}
\providecommand{\event}{NCL 2024} 

\usepackage{iftex}

\ifpdf
  \usepackage{underscore}         
  \usepackage[T1]{fontenc}        
\else
  \usepackage{breakurl}           
\fi

\usepackage{amsmath,amssymb,amsfonts}
\usepackage{amsthm}
\usepackage{xcolor}
\usepackage{nicefrac}

\newtheorem{theorem}{Theorem} 
\newtheorem{proposition}{Proposition}
\newtheorem{corollary}{Corollary}

\newcommand{\md}{\ensuremath{{\ominus}}}
\newcommand{\lra}{\ensuremath{{\leftrightarrow}}}
\newcommand{\sneg}{{\sim}}
\newcommand{\Ra}{\ensuremath{~{\Rightarrow}~}}
\newcommand{\V}{\mathcal{V}}
\newcommand{\M}{\mathcal{M}}
\newcommand{\F}{\mathcal{F}}
\newcommand{\RM}{\mathcal{RM}}
\newcommand{\Hi}{\mathcal{H}}
\newcommand{\axK}{\textsf{(K)}}
\newcommand{\axGL}{\textsf{(GL)}}
\newcommand{\axT}{\textsf{(T)}}
\newcommand{\axfo}{\textsf{(4)}}
\newcommand{\cpl}{{\bf CPL}}
\newcommand{\MP}{\textsf{MP}}
\newcommand{\ax}{\textsf{Ax}}

\DeclareSymbolFont{symbolsC}{U}{txsyc}{m}{n}
\DeclareMathSymbol{\strictif}{\mathrel}{symbolsC}{74}

\title{Modal Logics -- RNmatrices vs. Nmatrices}

\author{Marcelo E. Coniglio
\institute{University of Campinas\\ Campinas, Brazil}
\email{coniglio@unicamp.br}
\and
Pawe{\l} Paw{\l}owski
\institute{Ghent University\\
Ghent, Belgium}
\email{\quad pawel.pawlowski@ugent.be}
\and
Daniel Skurt
\institute{Ruhr University Bochum\\
Bochum, Germany}
\email{\quad daniel.skurt@rub.de}
}

\def\titlerunning{Modal Logics -- RNmatrices vs. Nmatrices}
\def\authorrunning{M.E. Coniglio, P. Pawlowski \& D. Skurt}

\begin{document}
\maketitle

\begin{abstract}

In this short paper we will discuss the similarities and differences between two semantic approaches to modal logics -- non-deterministic semantics and \emph{restricted} non-deterministic semantics. Generally speaking, both kinds of semantics are similar in the sense that they employ non-deterministic matrices as a starting point but differ significantly in the way extensions of the minimal modal logic {\bf M} are constructed.

Both kinds of semantics are many-valued and truth-values are typically expressed in terms of tuples of 0s and 1s, where each dimension of the tuple represents either truth/falsity, possibility/non-possibility, necessity/non-necessity etc.  And while non-deterministic semantics for modal logic offers an intuitive interpretation of the truth-values and the concept of modality, with restricted non-deterministic semantics are more general in terms of providing extensions of {\bf M}, including normal ones, in an uniform way.

On the example of three modal logics, {\bf MK}, {\bf MKT} and {\bf MKT4}, we will show the differences and similarities of those two approaches. Additionally, we will briefly discuss (current) restrictions of both approaches.
\end{abstract}

\section{Introduction} \label{sect:Intro}

We begin our study with the weakest system of modal logic --  {\bf M}.  This system is an expansion of classical propositional logic with a unary operator $\md$  and is characterized as follows: 

\begin{itemize}
	\item {\bf M} contains all (classical) tautologies 
	\item {\bf M} is closed under uniform substitution
	\item {\bf M} is closed under Modus Ponens
\end{itemize}

\noindent This starting point for investigating modal logics is not new. Logicians like Krister Segerberg \cite{segerberg1971}, David Makinson \cite{Makinson1971}, Heinrich Wansing \cite{wansing1989} and Lloyd Humberstone \cite{Humberstone2016} started their studies of modal logics with a similar weak system of modal logic, as well.\footnote{They either call it {\bf L$_0$}, {\bf PC} or {\bf S}. In the more recent \cite{Lukas2022} this system is called {\bf 0}. The main difference between their and our starting point is that they interpret  $\md$ from the beginning as a necessity operator.}

In the presentations for the smallest modal system by Segerberg, Makinson, Wansing or Humberstone the meaning of the modal operator $\md$ is not kept for all extensions of the smallest modal system. Since in practice, what happens is a shift of meaning for the operator $\md$. From no meaning in {\bf M} to a meaning constituted in possible worlds, where in all extensions, e.g. $\md A$ is true in a world iff $A$ is true in all accessible worlds.\footnote{If we interpret $ \md$ as necessity.} Similar things can be said about neighborhood frames or some versions of truth-maker semantics. These shifts are made rather abruptly in order to generate the needed behavior of the modal operator. In our approach we keep the meaning of $\md$ the same, and thus establish a uniform theory of modal operators.

This way of aiming at a uniform theory of modal operators is not new. In recent publications, cf.~\cite{Coniglio2019, OmoriSkurt22}, {\bf M} and some of its normal extensions were investigated as part of a larger discussion concerning non-deterministic semantics for non-normal modal logics, in the sense that the rule of necessitation is absent, and normal modal logics. There, the authors build upon the framework of non-deterministic semantics, which was systematically introduced by Arnon Avron and his collaborators, cf. \cite{AZSurvey}, but already used in the context of modal logics by Yuri Ivlev and John Kearns, cf. \cite{Ivlev88,Ivlev91}, \cite{Kearns81,Kearns89}, and further developed more recently for example in \cite{Coniglio2021, Lukas2022, OmoriSkurt16, OmoriSkurt2020, OmoriSkurt2021, pawlowski2022, PawlowskiSkurt2023, PawlowskiSkurt2024}. 

In this paper, we have a humble objective. We will present two different strategies of constructing semantics for modal logics via Nmatrices and via RNmatrices. In Section \ref{sec:minimal} we will introduce the minimal modal logic {\bf M} and show how it can be extended by either eliminating truth-values or non-determinacy in Section \ref{Nmat-exteMx} or by restricting the set of acceptable valuations in Section \ref{RNmat-exteM}\footnote{A largely extended version of that section is currently under review at another venue.}. This is then followed by Section \ref{sec:n-vs-rn}, where we 
will briefly compare both strategies, discuss some open problems and hint a future research of both semantics approaches.

\section{The minimal modal logic {\bf M}}\label{sec:minimal}

To start with, consider a modal propositional signature $\Sigma$ with unary connectives $\neg$ and $\md$ (classical negation and modality, respectively) and a binary connective $\to$ (material implication). Let $\V$ be a denumerable set of propositional variables $\V=\{p_0,p_1, \ldots\}$ and let $For(\Sigma)$ be the algebra of formulas over $\Sigma$ freely generated by $\V$. As usual, conjunction $\land$, disjunction $\vee$ and bi-implication \lra\ are defined from $\neg$ and $\to$ as follows: $A \land B := \neg(A \to \neg B)$, $A \vee B;=\neg A \to B$ and $A \lra B:=(A \to B) \land (B \to A)$.  Note that we could also take $\land$, $\lor$ and $\lra$ as primitive rather than defined connectives. However, due to the non-truth-functional nature of our semantics, presented below, this would require more care wrt the truth-tables and the formulation of later results. Hence, in order to keep our approach accessible to a broader audience, we decided to take smaller set of connectives as primitive.

In this section we consider a set of four-valued\footnote{The number of truth-values is not arbitrary but depends on the number of independent modal operators. In case we would like to consider two, three etc. independent modal operators, the number of truth-values would increase to eight, sixteen etc. values.} non-deterministic matrices (Nmatrices, for short) defined from swap structures (see for instance \cite[Ch. 6]{carnielli2016paraconsistent} and \cite{Coniglio2019a}) in which each truth-value is an ordered pair (or {\em snapshot}) $z=(z_1,z_2)$ in ${\bf2}^2$, for ${\bf2}=\{0,1\}$. Here, $z_1$ and $z_2$ represent, respectively, the truth value of $A$ and of $\md A$ for a given formula $A$ over $\Sigma$. This produces four truth-values $(1,0)$, $(1,1)$, $(0,1)$ and $(0,0)$. 
Let $V_4$ be the set of such truth-values. Accordingly, the set of designated values will be $D_4=\{z \in V_4 \ : \  z_1=1\}=\{(1,0),(1,1)\} = (1,\ast)$. On the other hand, the set of non-designated values is given as $ND_4=\{(0,0),(0,1)\}=(0,\ast)$.\footnote{Note that by $(1,\ast)$ and $(0,\ast)$ we mean sets of values rather than undefined values.}

Because of the intended meaning of the snapshots, i.e. classical operators should behave classically, negation and implication between snapshots are computed over {\bf 2} in the first coordinate, while the second one can takes an arbitrary value. That is:
\begin{center}
$
\begin{array}{rll}
\tilde{\neg}\, z&:=& (\sneg z_1,\ast);\\
z \,\tilde{\to}\, w &:=& (z_1 \Ra w_1,\ast) 
\end{array}
$
\end{center}
\noindent Here, $\sim$ and $\Ra$ denote the Boolean negation and the  implication in {\bf 2}. Observe that the second coordinate is arbitrary since at this moment $\md$ remains uninterpreted, i.e. 
there are no axioms ruling the value of $\md \neg A$ and the value of $\md(A \to B)$.

The interpretation of $\md$ is a multioperator which simply `reads' the second coordinate, while the second coordinate (corresponding to $\md\md A$) will be arbitrary at this point, as well:

\begin{center}
$
\begin{array}{rll}
\tilde{\md}\, z&:=& (z_2,\ast).
\end{array}
$
\end{center}

\noindent Let $\M= \langle V_4, D_4, \mathcal{O}\rangle$ be the obtained 4-valued Nmatrix, where $\mathcal{O}(\#)=\tilde{\#}$ for every connective $\#$ in $\Sigma$, with $\tilde{\#}: V_4 \longrightarrow \mathcal{P}(V_4)$.\footnote{I.e., truth-function for the connective assign non-empty sets of designated or non-designated values.} The truth-tables for $\M$ can be displayed as follows:\footnote{The Nmatrix semantics $\M$ for {\bf M} was already introduced by H. Omori and D. Skurt in~\cite{OmoriSkurt22}, but with a slight different interpretation of the truth-values.}

\begin{center}
\begin{tabular}{|c|c|c|c|c|}
\hline
 $\tilde{\to}$ & $(1,1)$   & $(1,0)$ & $(0,1)$ & $(0,0)$ \\
 \hline \hline
    $(1,1)$     & $(1,\ast)$   & $(1,\ast)$  & $(0,\ast)$ & $(0,\ast)$  \\ \hline
     $(1,0)$    & $(1,\ast)$   & $(1,\ast)$ 	& $(0,\ast)$ & $(0,\ast)$  \\ \hline
     $(0,1)$    & $(1,\ast)$   & $(1,\ast)$ 	& $(1,\ast)$ & $(1,\ast)$  \\ \hline
     $(0,0)$    & $(1,\ast)$   & $(1,\ast)$ 	& $(1,\ast)$ & $(1,\ast)$  \\ \hline
  \end{tabular}
\hspace{0.3cm}
\begin{tabular}{|c||c||c|} \hline
 $A$ & $\tilde{\neg} \, A$ & $\tilde{\md} \, A$ \\
 \hline \hline
     $(1,1)$  & $(0,\ast)$  & $(1,\ast)$   \\ \hline
     $(1,0)$  & $(0,\ast)$  & $(0,\ast)$    \\ \hline
     $(0,1)$  & $(1,\ast)$  & $(1,\ast)$  \\ \hline
     $(0,0)$  & $(1,\ast)$  & $(0,\ast)$   \\ \hline
  \end{tabular}
\end{center}

\noindent Now, let $\F$ be the set of all the valuations over the Nmatrix $\M$, such that  $v \in \F$ iff $v:For(\Sigma) \to V_4$ is a function satisfying the following properties:

\begin{itemize}
\item $v(\# A) \in \tilde{\#}\, v(A)$ for $\# \in \{\neg, \md\}$;
\item $v(A \to B) \in v(A) \, \tilde{\to}\, v(B)$.
\end{itemize}

\noindent The logic  {\bf M} generated by the Nmatrix $\M$ is then defined as follows:  $\Gamma \vDash_{\M} A$ iff, for every $v \in \F$: if $v(B) \in D_4$ for every $B \in \Gamma$ then $v(A) \in D_4$.

Alternatively, any valuation $v \in \F$ can be written as $v=(v_1,v_2)$ such that $v_1,v_2: For(\Sigma) \to {\bf2}$. 
Hence, $v(A)=(v_1(A),v_2(A))$ for every formula $A$. This means that, for all formulas $A$ and $B$:

\begin{itemize}
\item $v(A) \in D_4$ iff $v_1(A)=1$;
\item $v_1(\neg A)=\sneg v_1(A)$;
\item $v_1(\md A)= v_2(A)$;
\item $v_1(A \to B)= v_1(A) \Ra v_1(B)$.
\end{itemize}

\noindent The Hilbert calculus $\Hi$ for {\bf M} consists of the following axioms and a rule of inference.\footnote{Note that no axioms nor rules for $\md$ are given.}

\noindent 
{\small
\begin{minipage}{.5\textwidth}
\begin{align*}
& A \to (B \to A) \tag{\ax 1} \\
& (A \to (B \to C)) \to ((A \to B) \to (A \to C)) \tag{\ax 2}\\
& (\neg B \to \neg A) \to (A \to B)  \tag{\ax 3} \\
\end{align*}
\end{minipage}
\ 
\begin{minipage}{.5\textwidth}
\begin{align*}
& \displaystyle\frac{\ A \quad A {\to} B\ }{B} \tag{$\MP$} 
\end{align*}
\end{minipage}
}

\noindent We write $\Gamma \vdash_{\Hi} A$ if there is a sequence of formulas $B_1, \dots, B_n, A$, $n\geq 0$, such that every formula in the sequence either (i) belongs to $\Gamma$; (ii) is an axiom of $\Hi$; (iii) is obtained by (\MP) from formulas preceding it in sequence.

The following result is then easy to prove:

\begin{theorem} [Soundness and completeness of $\Hi$ w.r.t. $\M$] \label{complete-CPL}
For every $\Gamma \cup \{A\} \subseteq For(\Sigma)$ it holds: $\Gamma \vdash_{\Hi} A$ \ iff \  $\Gamma \vDash_{\M} A$.
\end{theorem}

Being the minimal modal logic, at first glance, {\bf M} seems to be nothing else than \cpl\ presented in a language with a modal operator \md\ without an interpretation. For instance, it does neither satisfy the axiom $\axK: \md(A \to B) \to (\md A \to \md B)$ nor the rule of necessitation. However, {\bf M} can not be characterized by a finite deterministic matrix, since any such characterization would designate a formula similar to the well-known Dugundji construction, which is of course not derivable in {\bf M}, cf. \cite{Lukas2022}.

Since $\md$ is supposed to represent any given modal operator (for instance, a possibility operator $\Diamond$) it should be expected that $\md$ has no fixed interpretation yet. But it can be shown that the nature of the modality, whether $\md$ can be interpreted as necessity, possibility, knowledge, obligation  etc., will strongly depend on our choice of axioms we want to be valid. However, $\md$ is not meaningless, since $\md A$ will be designated, iff $v_2(A) = 1$.

In the next sections we will extend the minimal modal logic {\bf M} with two modal axioms, $\md (A \to B) \to (\md A \to \md B)$ $\axK$, $\md A \to A$ $\axT$ and $\md A \to \md\md A$ \axfo, respectively, and later the rule of necessitation and thus show differences and similarities between two approaches for constructing non-deterministic semantics for modal logics. To this end, let $\Hi_{\textsf{K}}$ be the Hilbert calculus over $\Sigma$ obtained from $\Hi$ by adding axiom schema \axK. And let $\Hi_{\textsf{KT}}$ be the Hilbert calculus over $\Sigma$ obtained from $\Hi_{\textsf{K}}$ by adding axiom schema \axT. Furthermore, let $\Hi_{\textsf{KT4}}$ be the Hilbert calculus over $\Sigma$ obtained from $\Hi_{\textsf{KT}}$ by adding the axiom schema \axfo. The corresponding consequence relations $\vdash_{\Hi_{\textsf{K}}}$, $\vdash_{\Hi_{\textsf{KT}}}$ and $\vdash_{\Hi_{\textsf{KT4}}}$ are defined in similar manner than $\vdash_{\Hi}$

Before we continue, however, we will quickly show that these three axioms, \axK, \axT, \axfo, are not valid in {\bf M}. For $\axK$ consider a valuation such that $v(A) = (1,1)$ and $v(B) = (1,0)$, hence $v(\md A)$ is designated and $v(B)$ is non-designated, i.e. $\md A \to \md B$ is non-designated. Because $A \to B$ is designated and there is valuation such that $\md (A \to B)$ is designated, namely $v(A\to B) = (1,1)$, $\md (A \to B) \to (\md A \to \md B)$ is not valid in {\bf M}. As for $\axT$, consider a valuation such that $v(A) = (0,1)$. Then $v(\md A) = (1,1)$ or $(1,0)$. Hence, $\md A \to A$ will get a non-designated value. For $\axfo$ just take a valuation that assigns $\md A$ the value $(1,0)$.

\section{Nmatrices for {\bf MK}, {\bf MKT} and {\bf MKT4}}\label{Nmat-exteMx}

In this section, we will quickly recapitulate results from previous works, e.g. \cite{Coniglio2015,Coniglio2019a} or \cite{OmoriSkurt16,OmoriSkurt22} and show how to systematically develop non-deterministic semantics for certain extensions of {\bf M} by either eliminating some non-determinacy from the truth-tables or eliminating truth-values. 

\paragraph{{\bf MK}}

\noindent Let $\M_{\textsf{K}}= \langle V_4, D_4, \mathcal{O}\rangle$. The truth-tables for $\M_{\textsf{K}}$ can be displayed as follows:\footnote{We omit brackets for sets.}

\begin{center}
\begin{tabular}{|c|c|c|c|c|}
\hline
 $\tilde{\to}$ & $(1,1)$   & $(1,0)$ & $(0,1)$ & $(0,0)$ \\
 \hline \hline
    $(1,1)$     & $(1,\ast)$   & $(1,0)$  & $(0,\ast)$ & $(0,0)$  \\ \hline
     $(1,0)$    & $(1,\ast)$   & $(1,\ast)$ 	& $(0,\ast)$ & $(0,\ast)$  \\ \hline
     $(0,1)$    & $(1,\ast)$   & $(1,0)$ 	& $(1,\ast)$ & $(1,0)$  \\ \hline
     $(0,0)$    & $(1,\ast)$   & $(1,\ast)$ 	& $(1,\ast)$ & $(1,\ast)$  \\ \hline
  \end{tabular}
\hspace{0.3cm}
\begin{tabular}{|c||c||c|} \hline
 $A$ & $\tilde{\neg} \, A$ & $\tilde{\md} \, A$ \\
 \hline \hline
     $(1,1)$  & $(0,\ast)$  & $(1,\ast)$   \\ \hline
     $(1,0)$  & $(0,\ast)$  & $(0,\ast)$    \\ \hline
     $(0,1)$  & $(1,\ast)$  & $(1,\ast)$  \\ \hline
     $(0,0)$  & $(1,\ast)$  & $(0,\ast)$   \\ \hline
  \end{tabular}
\end{center}

\noindent Alternatively, we can calculate the value of $\Tilde{\to}$ by adding the following conditions to $\M$:
\begin{center}
$
\begin{array}{rll}
z \,\tilde{\to}\, w &:=& (x_1,x_2) = (x_1 = (z_1 \Ra w_1)~,~x_2 \leq z_2 \Ra w_2) 
\end{array}
\
\quad \text{ or }\qquad
\
v_2(A \to B) \leq v_2(A) \Ra v_2(B)
$
\end{center}

\paragraph{{\bf MKT}}

\noindent Now, let $V_3 = \{(z_1,z_2) \in V_4 \ : \ z_1 \geq z_2 \}$, i.e. $V_3 = V_4\backslash \{(0,1\}$. Accordingly, the set of designated values will be $D_3=\{z \in V_3 \ : \  z_1=1\} = D_4$. Then we can define $\M_{\textsf{KT}}= \langle V_3, D_3, \mathcal{O}\rangle$. The truth-tables for $\M_{\textsf{KT}}$ can be displayed as follows:

\begin{center}
\begin{tabular}{|c|c|c|c|c|}
\hline
 $\tilde{\to}$ & $(1,1)$   & $(1,0)$ & $(0,0)$ \\
 \hline \hline
    $(1,1)$     & $(1,\ast)$   & $(1,0)$  & $(0,0)$  \\ \hline
     $(1,0)$    & $(1,\ast)$   & $(1,\ast)$ 	& $(0,\ast)$  \\ \hline
     $(0,0)$    & $(1,\ast)$   & $(1,\ast)$ 	& $(1,\ast)$  \\ \hline
  \end{tabular}
\hspace{0.3cm}
\begin{tabular}{|c||c||c|} \hline
 $A$ & $\tilde{\neg} \, A$ & $\tilde{\md} \, A$ \\
 \hline \hline
     $(1,1)$  & $(0,\ast)$  & $(1,\ast)$   \\ \hline
     $(1,0)$  & $(0,\ast)$  & $(0,\ast)$    \\ \hline
     $(0,0)$  & $(1,\ast)$  & $(0,\ast)$   \\ \hline
  \end{tabular}
\end{center}

\noindent In this case, we do not need to change the definitions of the operations.

\paragraph{{\bf MKT4}}

\noindent Let $\M_{\textsf{KT4}}= \langle V_3, D_3, \mathcal{O}\rangle$. The truth-tables for $\M_{\textsf{KT4}}$ can be displayed as follows:

\begin{center}
\begin{tabular}{|c|c|c|c|c|}
\hline
 $\tilde{\to}$ & $(1,1)$   & $(1,0)$ & $(0,0)$ \\
 \hline \hline
     $(1,1)$    & $(1,\ast)$   & $(1,0)$  & $(0,0)$  \\ \hline
     $(1,0)$    & $(1,\ast)$   & $(1,\ast)$ 	& $(0,\ast)$  \\ \hline
     $(0,0)$    & $(1,\ast)$   & $(1,\ast)$ 	& $(1,\ast)$  \\ \hline
  \end{tabular}
\hspace{0.3cm}
\begin{tabular}{|c||c||c|} \hline
 $A$ & $\tilde{\neg} \, A$ & $\tilde{\md} \, A$ \\
 \hline \hline
     $(1,1)$  & $(0,\ast)$  & $(1,1)$   \\ \hline
     $(1,0)$  & $(0,\ast)$  & $(0,\ast)$    \\ \hline
     $(0,0)$  & $(1,\ast)$  & $(0,\ast)$   \\ \hline
  \end{tabular}
\end{center} 

\noindent Alternatively, we can calculate the value of $\Tilde{\md}$ by adding the following conditions to $\M$:

\begin{center}
$
\begin{array}{rll}
\tilde{\md}\, z &:=& (x_1,x_2) = (x_1 = z_2~,~x_2 \geq z_2) 
\end{array}
\
\quad \text{ or }\qquad
\
v_2(\md A) \geq v_2(A)
$
\end{center}

\bigskip

\noindent Let $\mathsf{Ax} \in \{\mathsf{K, KT, KT4}\}$, then we have the following results:

\begin{theorem} [Soundness and completeness of $\Hi_{\textsf{Ax}}$ w.r.t. the Nmatrix $\M_{\textsf{Ax}}$] \label{sound-compl-nmatrix}
Let $\Gamma \cup \{A\} \subseteq For(\Sigma)$. Then: $\Gamma \vdash_{\Hi_{\textsf{Ax}}} A$  iff $\Gamma \vDash_{\M_{\textsf{Ax}}} A$.
\end{theorem}

\noindent We omit the proofs, as they can be found in detail in previous publications.

\section{RNmatrices for {\bf MK}, {\bf MKT} and {\bf MKT4}}\label{RNmat-exteM}

In this section, we will go into a little more detail, since we do not expect readers to be familiar with what in \cite{Coniglio2021a} was called  {\em restricted Nmatrices} (RNmatrices). In short, the set $\F$ of valuations over the Nmatrix $\M$ will be restricted to specific subsets $\F' \subseteq \F$ with the aim of satisfying certain modal axiom(s).
In particular, we will consider RNmatrices of the form $\RM=\langle \M, \F'\rangle$ such that $\F' \subseteq \F$. I.e., each axiom that extends {\bf M} will restrict the set of acceptable valuation.

\paragraph{{\bf MK}}

Consider
$$\axK:\md(A \to B) \to (\md A \to \md B)$$ 
Then, a valuation $v \in \F$ satisfies \axK\ iff, $v_1(\md(A \to B) \to (\md A \to \md B))=1$ for every $A,B$ iff 
$v_1(\md(A \to B)) \Ra (v_1(\md A) \Ra v_1(\md B))=1$ for every $A,B$ iff $v_1(\md(A \to B)) \leq v_1(\md A) \Ra v_1(\md B)$ for every $A,B$ iff $v_2(A \to B) \leq v_2(A) \Ra v_2(B)$ for every $A,B$. Hence, the logic ${\bf MK}$ satisfying axiom \axK\ is characterized by the RNmatrix $\RM_{\textsf{K}}=\langle \M, \F_{\textsf{K}}\rangle$ such that 
$$\F_{\textsf{K}}=\{v \in \F \ : \ v_2(A \to B) \leq v_2(A) \Ra v_2(B) \ \mbox{ for every $A,B$} \}.$$

\paragraph{{\bf MKT}} Consider
$$\axT:\md A  \to A$$ 
Then, a valuation $v \in \F$ satisfies \axT\ iff, $v_1(\md A  \to A)=1$ for every $A$ iff 
$v_1(\md A)  \Ra v_1(A)=1$ for every $A$ iff $v_1(\md A)  \leq v_1(A)$ for every $A$ iff $v_2(A)  \leq v_1(A)$ for every $A$. Hence, the logic ${\bf MKT}$ satisfying axioms \axK and \axT\ is characterized by the RNmatrix $\RM_{\textsf{KT}}=\langle \M, \F_{\textsf{KT}}\rangle$ such that 
$$\F_{\textsf{KT}}=\{v \in \F \ : \ v_2(A \to B) \leq v_2(A) \Ra v_2(B) \  \mbox{and} \ v_2(A)  \leq v_1(A) \mbox{ for every $A,B$} \}.$$

\paragraph{{\bf MKT4}} Consider
$$\axfo:\md A  \to \md\md A$$ 
Then, a valuation $v \in \F$ satisfies \axfo\ iff, $v_1(\md A  \to \md\md A)=1$ for every $A$ iff 
$v_1(\md A)  \Ra v_1(\md\md A)=1$ for every $A$ iff $v_1(\md A)  \leq v_1(\md\md A)$ for every $A$ iff $v_2(A)  \leq v_2(\md A)$ for every $A$. Hence, the logic ${\bf MKT4}$ satisfying axioms \axK, \axT and \axfo\ is characterized by the RNmatrix $\RM_{\textsf{KT4}}=\langle \M, \F_{\textsf{KT4}}\rangle$ such that 
$$\F_{\textsf{KT4}}=\{v \in \F \ : \ v_2(A \to B) \leq v_2(A) \Ra v_2(B) \  \mbox{and} \ v_2(A)  \leq v_1(A) \  \mbox{and} \ v_2(A)  \leq v_2(\md A) \mbox{ for every $A,B$} \}.$$  

\noindent Each of the generated RNmatrices  will be {\em structural}, that is, the set $\F'$ is required to be closed under substitutions:
if $v \in \F'$ and $\rho$ is a substitution over $\Sigma$ then $v \circ \rho \in \F'$.
As proved in \cite{Coniglio2021a}, any structural RNmatrix generates a Tarskian and structural consequence relation defined as expected: $\Gamma \vDash_{\RM} A$ iff, for every $v \in \F'$: if $v(B) \in D_4$ for every $B \in \Gamma$ then $v(A) \in D_4$.

Recall, a {\em substitution} over the signature $\Sigma$ is a function $\sigma:\mathcal{V} \to For(\Sigma)$. Since $For(\Sigma)$ is an absolutely free algebra, each $\sigma$ can be extended to a unique endomorphism in $For(\Sigma)$ (which will be also denoted by $\sigma$). That is, $\sigma:For(\Sigma) \to For(\Sigma)$ is such that $\sigma(\# A)=\#\sigma(A)$ for $\# \in \{\neg,\md\}$, and $\sigma(A \to B)=\sigma(A) \to \sigma(B)$. The set of substitutions over $\sigma$ (seen as endomorphisms in $For(\Sigma)$) will be denoted by $Subs(\Sigma)$.

Clearly, $\RM_{\textsf{K}}$, $\RM_{\textsf{KT}}$ and $\RM_{\textsf{KT4}}$ are structural, hence generate a Tarskian and structural consequence relation $\vDash_{\RM_{\textsf{K}}}$, $\vDash_{\RM_{\textsf{KT}}}$, $\vDash_{\RM_{\textsf{KT4}}}$. 

E.g., let $\rho$ be a substitution and let $v \in \F_{\textsf{K}}$. 
Observe that $v \circ \rho=(v_1 \circ \rho,v_2\circ \rho)$. Then, for every $A,B$: $v_2\circ \rho(A \to B)= v_2(\rho(A \to B))= v_2(\rho(A) \to \rho(B)) \leq v_2(\rho(A)) \Ra v_2((\rho(B))= v_2\circ\rho(A) \Ra v_2\circ \rho(B)$. Hence $v \circ \rho \in \F_{\textsf{K}}$.

\bigskip

\noindent We will now sketch soundness and completeness results for $\Hi_{\textsf{K}}$, the proofs for $\Hi_{\textsf{KT}}$ and $\Hi_{\textsf{KT4}}$ follow the same structure. More details for soundness and completeness of Hilbert calculi wrt RNmatrices can be found for example \cite{Coniglio2021a}. To this end, we will make use of well-known definitions.

Recall that, given a Tarskian and finitary logic {\bf L}, a set of formulas $\Delta$ is said to be $A$-saturated (where $A$ is a formula) if $\Delta \nvdash_{\bf L} A$ but $\Delta,B \vdash_{\bf L} A$ for every formula $B$ such that $B \notin \Delta$. If $\Delta$ is $A$-saturated then it is a closed theory, that is: $\Delta \vdash_{\bf L} B$ iff $B \in \Delta$. It is well-known that, in any Tarskian and finitary logic {\bf L}, if $\Gamma   \nvdash_{\bf L} A$ then there exists an $A$-saturated set $\Delta$ such that $\Gamma \subseteq \Delta$. Since the logic generated by $\Hi_{\textsf{K}}$ is Tarskian and finitary, it has this property.

\begin{proposition} \label{A-sat-HK}
Let $\Delta$ be an $A$-saturated set in  $\Hi_{\textsf{K}}$. Then, for every formulas $A,B$:\\
(1) $\neg A \in \Delta$ iff $A \not\in \Delta$;\\
(2) $A \to B \in \Delta$ iff either $A \not\in \Delta$ or $B \in \Delta$;\\
(3) if $\md(A \to B) \in \Delta$ and $\md A \in \Delta$  then $\md B \in \Delta$. 
\end{proposition}
\begin{proof}
Immediate, by definition of $\Hi_{\textsf{K}}$ and the fact that $\Delta$ is a closed theory.
\end{proof}

\begin{corollary} \label{bival-sat} Let $\Delta$ be an $A$-saturated set in  $\Hi_{\textsf{K}}$. Then, $B \in \Delta$ iff $v(B) \in D_4$ for some $v\in \mathcal{F}_{\textsf{K}}$.
\end{corollary}
\begin{proof}
It is immediate from Proposition~\ref{A-sat-HK}.\footnote{More detailed proofs regarding soundness and completeness for RNmatrices can be found in \cite{Coniglio2021a}}
\end{proof}

\begin{theorem} [Soundness and completeness of $\Hi_{\textsf{K}}$] \label{complete-HK-RN}
Let $\Gamma \cup \{A\} \subseteq For(\Sigma)$. Then: $\Gamma \vdash_{\Hi_{\textsf{K}}} A$  iff $\Gamma \vDash_{\RM_{\textsf{K}}} A$.
\end{theorem}

\begin{proof} \ \\
{\em (Soundness)}: It is an easy exercise to show that every axiom of $\Hi_{\textsf{K}}$ is valid w.r.t. $\RM_{\textsf{K}}$ and that $\MP$ preserves the designated values.
Hence, by induction on the length of a derivation in $\Hi_{\textsf{K}}$  of $A$ from $\Gamma$, it is easy to see the following: $\Gamma \vdash_{\Hi_{\textsf{K}}} A$  implies that $\Gamma \vDash_{\RM_{\textsf{K}}} A$.\\
{\em (Completeness)}: Suppose that $\Gamma \nvdash_{\Hi_{\textsf{K}}} A$. As observed above, there exists an $A$-saturated set $\Delta$ such that $\Gamma \subseteq \Delta$. By Corollary~\ref{bival-sat} there is a valuation $v\in \mathcal{F}_{\textsf{K}}$ such that 
$v(B)=1$ for every $B \in \Gamma$ but $v(A)=0$. This shows that $\Gamma \nvDash_{\RM_{\textsf{K}}} A$.
\end{proof}

\section{Nmatrices vs. RNmatrices}\label{sec:n-vs-rn}

In this short article we sketched out two general semantical framework for modal logics that do not rely on the notion of possible worlds. Both frameworks adopt non-deterministic semantics in a creative way in order to construct alternative semantics for systems with modalities, but have different strategies for constructing extensions of the minimal modal logic {\bf M}. While the status of {\bf M} as a modal logic itself, can be discussed, in this section we will briefly compare both approaches and consider some similarities and differences of them.

We will start by showing how to extend the systems discussed in this article with the rule of necessitation in order to construct normal modal logics. This can be done for Nmatrices and RNmatrices in a similar manner. Then, we will show some limitations of both approaches. 
Finally, we will conclude with remarks on some philosophical issues concerning both approaches.

\subsection{Normal modal logics}

The systems presented before, even though we call them modal logics, are generally not received as such, with the reason being that rules for the modal operators are not present. And while we will not in full detail discuss our rationale behind our terminology, we can certainly provide the technical means such that the systems can be extended by (any) modal rules. 

As an example, we will present the changes to the semantics needed for validating the rule of necessitation ({\sf N}). 

Let $\mathsf{Ax} \in \{\mathsf{K, KT, KT4}\}$. Then $\Hi_\ax^{\sf N}$ is the Hilbert calculus obtained from  $\Hi_\ax$ by adding the necessitation  inference rule (where $A$ is a propositional variable):  $$\frac{A}{\md A} \quad({\sf N})$$
A formula $B$ is derivable in  $\Hi_\ax^{\sf N}$ if there exists a finite sequence of formulas $A_1 \ldots A_m$ such that $A_m=B$ and, for every $1\leq i\leq m$, either $A_i$ is an instance of an axiom, or it follows from $A_j=A_k \to A_i$ and $A_k$ by \MP), for some $j,k<i$, or $A_i=\md A_j$ follows from $A_j$ (for some $j<i$) by the {\sf N}-rule. Finally, $A$ is derivable from $\Gamma$ in  $\Hi_\ax^{\sf N}$ if either $A$ is derivable in  $\Hi_\ax^{\sf N}$, or  $B_1  \to (B_2 \to ( \ldots \to (B_k \to A) \ldots )$ is derivable in  $\Hi_\ax^{\sf N}$ for some nonempty finite set $\{B_1,\ldots,B_k\} \subseteq \Gamma$.

First observe that neither Nmatrices nor RNmatrices will capture the behavior of the {\sf N}-rule. To see this, take for example any classical tautology, e.g. $A \to A$. Due to both kinds of semantics, this formula will receive the value $(1,0)$, for some valuation $v$. But then for this valuation we have $v(\md (A\to A)) \notin D_4$. Hence the {\sf N}-rule is not valid.

The key idea for the validity of that rule was given by John Kearns in \cite{Kearns81}. There he restricted the set of acceptable valuations by a simple strategy. If a formula $A$ receives a designated value wrt all possible valuation, than $\md A$ will also receive a designated value. I.e., for the formula $A$ in question, all valuations $v$, such that $v(A) = (1,0)$, will be eliminated from the set of acceptable valuations.  
The original idea of Kearns was later put in the modern context of non-deterministic semantics, see for example \cite{Coniglio2015, OmoriSkurt16}. Following the results from~\cite{Coniglio2015,Kearns81,OmoriSkurt16}, we define the set of valuations over the RNmatrix $\RM_{\ax}=\langle \M, \F_\ax\rangle$ as follows:

\begin{itemize}
    \item $\F^0_\ax = \F_\ax$
    \item $\F^{m+1}_\ax = \big\{v \in \F^{m}_\ax \ : \ \forall B\in For(\Sigma),  \text{ if }\forall w\in  \F^{m}_\ax(w_1(B) = 1)  \text{ then } v_2(B)=1  \big\}$
    \item $\F^{\sf N}_\ax = \overset{\infty}{\underset{m=0}{\bigcap}} \F^m_\ax$
\end{itemize}
Observe that $\F^{\sf N}_\ax$ coincides with the original definition of level valuations introduced by Kearns and therefore can also applied to Nmatrices by using instead of the restricted set of valuation $\F_\ax$, the set of all valuation $\F$.

We then define a new semantical consequence as follows: 

\begin{enumerate}
\item $A$ is valid, denoted by $\models_{\RM_{\ax}} A$, if $v_1(A)=1$ for every $v \in \F^{N}_\ax$.
\item $A$ is a semantical consequence of $\Gamma$, denoted by $\Gamma \models_{\RM_{\ax}} A$, if either $A$ is valid, or  $B_1  \to (B_2 \to ( \ldots \to (B_k \to A) \ldots )$ is valid for some nonempty finite set $\{B_1,\ldots,B_k\} \subseteq \Gamma$.
\end{enumerate}

For both, Nmatrices and RNmatrices, soundness and completeness results can be obtained in a straight forward manner, by adapting the results from \cite{Coniglio2015,Kearns81,OmoriSkurt16}.

Recent unpublished results by the authors also show that this technique can be adapted for any (global) modal rule, and thus offering semantics for well-known non-normal modal logics as well. In this regard, both kinds of semantics offer a unifying framework for a various range of (non)-normal modal logics. But while the technique of level valuations applied to Nmatrices might seem ad hoc, it is a generalization of restricting the valuations for RNmatrices. In fact, one could interpret the restriction method for RNmatrices as a local restriction of valuations and the level valuation technique as a global restriction of valuation. 

\smallskip

Finally, we note that it is possible to expand the language with additional modal operators. Semantically this can be done by adding more dimensions to the truth-values. Rather than pairs, truth-values would then be $n$-tuples, depending on the number of additional modal operators. 

For example, we can consider a bimodal version of the minimal modal logic {\bf M}, namely the minimal bimodal logic {\bf M2}. This logic is defined over a signature $\Sigma_2$ obtained from $\Sigma$ by replacing $\md$ with two modal operators, which will be denoted by $\md_1$ and $\md_2$. As expected, the snapshots are now triples $z=(z_1,z_2,z_3)$ over {\bf 2} in which each coordinate represents a possible truth-value for the formulas $A$, $\md_1 A$ and $\md_2 A$, respectively. Hence, eight truth-values are $V_8 = \{(z_1,z_2,z_3) \ : \ z_1,z_2,z_3 \in \{1,0\}\}$, with $D_8=\{z \in V_8 \ : \  z_1=1\}$ being the set of designated values. 

The definition of the multioperators over $V_8$ interpreting the connectives of $\Sigma_2$ is a natural generalization  of the 4-valued case:

$
\begin{array}{rll}
\tilde{\neg}\, z&:=& (\sneg z_1,\ast,\ast);\\
\tilde{\md}_1\, z&:=& (z_2,\ast,\ast);\\
\tilde{\md}_2\, z&:=& (z_3,\ast,\ast);\\
z \,\tilde{\to}\, w &:=& (z_1 \Ra w_1,\ast,\ast). 
\end{array}
$

\noindent Let $\M_2= \langle V_8, D_8, \mathcal{O}_2\rangle$ be the obtained 8-valued Nmatrix, where $\mathcal{O}_2(\#)=\tilde{\#}$ for every connective $\#$ in $\Sigma_2$. Thus, the truth-tables for $\M_2$ are the following (for reasons of limited space we only present the truth-tables for $\md_1$ and $\md_2$): 

\begin{center}
$
\begin{array}{|c||c|c|c|c|c|c|c|c|}\hline
A & (1,1,1) & (1,1,0) & (1,0,1)& (1,0,0)& (0,1,1)& (0,1,0)& (0,0,1)& (0,0,0) \\\hline \hline
\md_1 A & (1,\ast,\ast) & (1,\ast,\ast) & (0,\ast,\ast)& (0,\ast,\ast)& (1,\ast,\ast)& (1,\ast,\ast)& (0,\ast,\ast)& (0,\ast,\ast) \\\hline \hline
\md_2 A & (0,\ast,\ast)& (0,\ast,\ast)& (1,\ast,\ast)& (1,\ast,\ast)& (0,\ast,\ast)& (0,\ast,\ast)&(1,\ast,\ast)&(1,\ast,\ast)\\\hline
\end{array}
$
\end{center}

While extensions of {\bf M2} wrt Nmatrices have been discussed at large for example in \cite{Coniglio2021, Lukas2022, OmoriSkurt16, OmoriSkurt2020, OmoriSkurt2021, pawlowski2022, PawlowskiSkurt2023, PawlowskiSkurt2024}, we present restricions for some axioms wrt RNmatrices:

\begin{center}
\begin{tabular}{|c|c|} \hline
Axiom & Restrictions \\
 \hline \hline
$\md_2 A \leftrightarrow \neg\md_1 \neg A$ &  $v_2(\neg A) = {\sim}v_3(A)$ \\[1mm] \hline
$\md_1 A \leftrightarrow \neg\md_2 \neg A$ &  $v_3(\neg A) = {\sim}v_2(A)$ \\[1mm] \hline
$\md_1 A \to \md_2 A$ &  $v_2(A) \leq v_3(A)$ \\[1mm] \hline
$A \to \md_1\md_2 A$ & $v_1(A) \leq v_2(\md_2 A)$ \\[1mm] \hline
$\md_2 A \to \md_1 \md_2 A$ & $v_3(A) \leq v_2(\md_2 A)$ \\[1mm] \hline
$\md_2\md_1 A \to \md_1\md_2 A$ & $v_3(\md_1 A) \leq v_2(\md_2 A)$ \\[1mm] \hline
  \end{tabular}
\end{center}

We just note in passing that it is not known, whether the last axiom, $\md_2\md_1 A \to \md_1\md_2 A$, can be represented in terms of Nmatrices at all. But both approaches to modal logics are certainly capable of producing semantics for a wide range of normal modal logics.

\subsection{Scope and limits of the methods}

Semantics for modal logics via Nmatrices share the heuristics of systematically eliminating semantical values or non-determinacy from non-deter\-minis\-tic truth-tables to validate desired axioms. This approach is successful in providing a uniform semantics for a broad class of normal and non-normal modal systems, even for some systems that lack possible worlds semantics, cf.  \cite{OmoriSkurt2020}. Hence, the proposed framework seems not only more general but also conceptually conservative since the meaning of modal operators was kept uniformly.

However, the technique of eliminating values or non-determinacy has its limitations. It became apparent that not all modal axioms could be straightforwardly represented in a non-deterministic truth-table format, such as the Gödel-Löb axiom $\md(\md A \to A) \to \md A$ \axGL~or other non-Sahlquist formulas. The possibility of providing Nmatrices for such formulas remains uncertain.

RNmatrices on the other hand are much more flexible in this regard. In fact, as shown in \cite{Coniglio2021a}, RNmatrices are stronger than Nmatrices. For example, by analysis similar to the one given for $\axK$, $\axT$ or $\axfo$, it is immediate to see that the restriction on the valuation imposed by this axiom is the following:
 $v_2(\md A \to A) \leq  v_2(A)$. 

However, we are aware of the fact that even RNmatrices are not yet as flexible as Kripke semantics regarding some properties. For example, we have not discussed axiom systems with an infinite number of axioms. While the construction method of RNmatrices for extensions of \textbf{M} might give us some arguments that, at least, recursively defined infinite sets of axioms might be expressed in terms of RNmatrices, we might discuss infinite axiom systems, but leave this as a project for future work.

As binary modal operators, like strict implication, the situation is slightly more complicated. For example, in case of strict implication, it seems, the corresponding Kripke semantics implicitly uses a global rule in the definition of the operators, which is something that cannot be expressed in terms of RNmatrices alone, at least not in straight forward manner. We could think of defining strict implication $\strictif$ in terms of $\to$, as follows: $v(A\strictif B) = v(\Box(A\to B))$ and depending on our semantics for $\Box$, we could define different notions of strictness. However, without globally restricting the set of all valuations, none of the sentences $A\strictif B$ would be a tautology. That is because no sentence $A\to B$ would will be assigned the value $(1,1)$ for all valuations. Again, for the moment, we will leave the question of how to define $n$-ary modal operators open. 

Another limitation of our approach at this moment is a property that is called analyticity. In short, if analyticity holds any partial valuation which seems to refute a given formula can be extended to a full valuation (which necessarily refutes that formula too). For example in \cite{OmoriSkurt16} it was shown that modal logics defined in terms of Nmatrices with global modal rules do not enjoy this property. It is obvious that the failure of analyticity carries over to RNmatrices with global modal rules. Since the failure of analyticity is related to decidability, it seems our presented semantics for modal logics with global modal rules are not decidable. It should of course be mentioned, that in the absence of such global modal rules it can be shown that our semantics are indeed decidable. Needless to say there is gleam of hope. In more recent publications, cf. \cite{Lukas2022} and \cite{lav:zoh:22}, it was shown that by a slight adjustment of the level-valuations technique it is possible to regain decidability. The results were proven for the normal modal logics \textbf{K}, \textbf{KT} and \textbf{S4} expressed in terms of Nmatrices. Recently, in~\cite{coniglio2023} the method for \textbf{S4} obtained in~\cite{Lukas2022} was adapted to obtain a decidable RNmatrix with level valuations for intuitionistic propositional logic. It is therefore only a matter of time to prove similar results for other modal (or non-classical) logics with global inference rules.

\subsection{Philosophical Remarks}

The two approaches we presented lead to semantics with sound and complete axiom systems with global modal rules. In that sense, at least with the addition of modal rules, we are justified to claim that we are actually doing modal logics. However, in the absence of such global modal rules, it seems, at the very least, questionable what the status of our operator $\md$ might be. Surely, we can define restrictions on the set of valuations that validate well-known modal formulas. But this is not yet an argument in favor of the modal nature of $\md$. We could furthermore think of concrete well-studied modal systems such as systems of epistemic of deontic logic, where the rule of necessitation is the source of some paradoxes and therefore not unrestrictedly valid. But even in such systems other global modal rules are present, such as congruentiality.

There are logics, called hyperintensional logics, for which even congruentiality fails to hold, cf. \cite{sep-hyperintensionality}, and our approach is certainly able to capture such logics, as well, but we should be very clear, that we are not discussing any particular modal operator. Instead, what can be said in favor of our approach, we are able to capture a multitude of different modal concepts under one and the same umbrella -- RNmatrices (or to a lesser extent Nmatrices) with or without global modal rules. Whether this will lead to a new understanding of the concept of modality remains to be open, and needs to be part of a larger investigation and discussion in the future.

Finally, we just remark in passing that the obvious elephant in the room, namely the correlation between Kripke semantics and Nmatrices/ RNmatrices has so far not yet been thoroughly investigated, even though it seems to be a captivating topic. We will leave this subject for future research.

\nocite{*}
\bibliographystyle{eptcs}
\bibliography{genericM}

\begin{thebibliography}{10}
\providecommand{\bibitemdeclare}[2]{}
\providecommand{\surnamestart}{}
\providecommand{\surnameend}{}
\providecommand{\urlprefix}{Available at }
\providecommand{\url}[1]{\texttt{#1}}
\providecommand{\href}[2]{\texttt{#2}}
\providecommand{\urlalt}[2]{\href{#1}{#2}}
\providecommand{\doi}[1]{doi:\urlalt{https://doi.org/#1}{#1}}
\providecommand{\eprint}[1]{arXiv:\urlalt{https://arxiv.org/abs/#1}{#1}}
\providecommand{\bibinfo}[2]{#2}

\bibitemdeclare{article}{Avron2007}
\bibitem{Avron2007}
\bibinfo{author}{Arnon \surnamestart Avron\surnameend} (\bibinfo{year}{2007}):
  \emph{\bibinfo{title}{Non-deterministic semantics for logics with a
  consistency operator}}.
\newblock {\slshape \bibinfo{journal}{International Journal of Approximate
  Reasoning}} \bibinfo{volume}{45}(\bibinfo{number}{2}), pp.
  \bibinfo{pages}{271--287}, \doi{10.1016/j.ijar.2006.06.011}.

\bibitemdeclare{incollection}{AL2001}
\bibitem{AL2001}
\bibinfo{author}{Arnon \surnamestart Avron\surnameend} \& \bibinfo{author}{Iddo
  \surnamestart Lev\surnameend} (\bibinfo{year}{2001}):
  \emph{\bibinfo{title}{Canonical propositional {G}entzen-type systems}}.
\newblock In \bibinfo{editor}{R.~\surnamestart Gor\'e\surnameend},
  \bibinfo{editor}{A.~\surnamestart Leitsch\surnameend} \&
  \bibinfo{editor}{T.~\surnamestart Nipkow\surnameend}, editors: {\slshape
  \bibinfo{booktitle}{Proceedings of the First International Joint Conference
  on Automated Reasoning (IJCAR ’01)}}, {\slshape \bibinfo{series}{Lecture
  Notes in Artificial Intelligence}} \bibinfo{volume}{2083}, pp.
  \bibinfo{pages}{529--544}, \doi{10.1007/3-540-45744-5_45 publisher =
  {Springer-Verlag}}.

\bibitemdeclare{article}{AL2005}
\bibitem{AL2005}
\bibinfo{author}{Arnon \surnamestart Avron\surnameend} \& \bibinfo{author}{Iddo
  \surnamestart Lev\surnameend} (\bibinfo{year}{2005}):
  \emph{\bibinfo{title}{{Non-Deterministic Multiple-valued Structures}}}.
\newblock {\slshape \bibinfo{journal}{Journal of Logic and Computation}}
  \bibinfo{volume}{15}(\bibinfo{number}{3}), pp. \bibinfo{pages}{241--261},
  \doi{10.1093/logcom/exi001}.

\bibitemdeclare{incollection}{AZSurvey}
\bibitem{AZSurvey}
\bibinfo{author}{Arnon \surnamestart Avron\surnameend} \& \bibinfo{author}{Anna
  \surnamestart Zamansky\surnameend} (\bibinfo{year}{2011}):
  \emph{\bibinfo{title}{Non-Deterministic Semantics for Logical Systems}}.
\newblock In: {\slshape \bibinfo{booktitle}{Handbook of Philosophical Logic}},
  \bibinfo{volume}{16}, \bibinfo{publisher}{Springer}, pp.
  \bibinfo{pages}{227--304}, \doi{10.1007/978-94-007-0479-4_4}.

\bibitemdeclare{incollection}{sep-hyperintensionality}
\bibitem{sep-hyperintensionality}
\bibinfo{author}{Francesco \surnamestart Berto\surnameend} \&
  \bibinfo{author}{Daniel \surnamestart Nolan\surnameend}
  (\bibinfo{year}{2023}): \emph{\bibinfo{title}{{Hyperintensionality}}}.
\newblock In \bibinfo{editor}{Edward~N. \surnamestart Zalta\surnameend} \&
  \bibinfo{editor}{Uri \surnamestart Nodelman\surnameend}, editors: {\slshape
  \bibinfo{booktitle}{The {Stanford} Encyclopedia of Philosophy}},
  \bibinfo{edition}{{W}inter 2023} edition, \bibinfo{publisher}{Metaphysics
  Research Lab, Stanford University}.

\bibitemdeclare{book}{blackburnetal2001}
\bibitem{blackburnetal2001}
\bibinfo{author}{Patrick \surnamestart Blackburn\surnameend},
  \bibinfo{author}{Maarten \surnamestart De~Rijke\surnameend} \&
  \bibinfo{author}{Yde \surnamestart Venema\surnameend} (\bibinfo{year}{2001}):
  \emph{\bibinfo{title}{Modal logic}}.
\newblock \bibinfo{volume}{53}, \bibinfo{publisher}{Cambridge University
  Press}, \doi{10.1017/CBO9781107050884}.

\bibitemdeclare{book}{carnielli2016paraconsistent}
\bibitem{carnielli2016paraconsistent}
\bibinfo{author}{Walter~Alexandre \surnamestart Carnielli\surnameend} \&
  \bibinfo{author}{Marcelo~Esteban \surnamestart Coniglio\surnameend}
  (\bibinfo{year}{2016}): \emph{\bibinfo{title}{Paraconsistent logic:
  Consistency, contradiction and negation}}.
\newblock {\slshape \bibinfo{series}{Logic, Epistemology, and the Unity of
  Science}}~\bibinfo{volume}{40}, \bibinfo{publisher}{Springer},
  \doi{10.1007/978-3-319-33205-5_5}.

\bibitemdeclare{article}{Coniglio2015}
\bibitem{Coniglio2015}
\bibinfo{author}{Marcelo~E. \surnamestart Coniglio\surnameend},
  \bibinfo{author}{Luis \surnamestart Fari{\~n}as Del~Cerro\surnameend} \&
  \bibinfo{author}{Newton~M. \surnamestart Peron\surnameend}
  (\bibinfo{year}{2015}): \emph{\bibinfo{title}{Finite Non-Deterministic
  Semantics for Some Modal Systems}}.
\newblock {\slshape \bibinfo{journal}{Journal of Applied Non-Classical Logics}}
  \bibinfo{volume}{25}(\bibinfo{number}{1}), pp. \bibinfo{pages}{20--45},
  \doi{10.1080/11663081.2015.1011543}.

\bibitemdeclare{article}{Coniglio2019}
\bibitem{Coniglio2019}
\bibinfo{author}{Marcelo~E. \surnamestart Coniglio\surnameend},
  \bibinfo{author}{Luis \surnamestart Fari{\~n}as Del~Cerro\surnameend} \&
  \bibinfo{author}{Newton~M. \surnamestart Peron\surnameend}
  (\bibinfo{year}{2019}): \emph{\bibinfo{title}{{Modal logic with
  non-deterministic semantics: Part {I}—Propositional case}}}.
\newblock {\slshape \bibinfo{journal}{Logic Journal of the IGPL}}, pp.
  \bibinfo{pages}{281--315}, \doi{10.1093/jigpal/jzz027}.

\bibitemdeclare{article}{Coniglio2021}
\bibitem{Coniglio2021}
\bibinfo{author}{Marcelo~E. \surnamestart Coniglio\surnameend},
  \bibinfo{author}{Luis \surnamestart Fari{\~n}as Del~Cerro\surnameend} \&
  \bibinfo{author}{Newton~M. \surnamestart Peron\surnameend}
  (\bibinfo{year}{2021}): \emph{\bibinfo{title}{{Modal Logic With
  Non-Deterministic Semantics: Part {II}—Quantified Case}}}.
\newblock {\slshape \bibinfo{journal}{Logic Journal of the IGPL}}, pp.
  \bibinfo{pages}{695--727}, \doi{10.1093/jigpal/jzab020}.

\bibitemdeclare{article}{Coniglio2019a}
\bibitem{Coniglio2019a}
\bibinfo{author}{Marcelo~E. \surnamestart Coniglio\surnameend} \&
  \bibinfo{author}{Ana~Claudia \surnamestart Golzio\surnameend}
  (\bibinfo{year}{2019}): \emph{\bibinfo{title}{Swap structures semantics for
  {I}vlev-like modal logics}}.
\newblock {\slshape \bibinfo{journal}{Soft Computing}}
  \bibinfo{volume}{23}(\bibinfo{number}{7}), pp. \bibinfo{pages}{2243--2254},
  \doi{10.1007/s00500-018-03707-4}.

\bibitemdeclare{article}{Coniglio2021a}
\bibitem{Coniglio2021a}
\bibinfo{author}{Marcelo~E. \surnamestart Coniglio\surnameend} \&
  \bibinfo{author}{Guilherme~V. \surnamestart Toledo\surnameend}
  (\bibinfo{year}{2021}): \emph{\bibinfo{title}{Two Decision Procedures for da
  {C}osta's ${C}_n$ Logics Based on {R}estricted {N}matrix Semantics}}.
\newblock {\slshape \bibinfo{journal}{Studia Logica}}
  \bibinfo{volume}{110}(\bibinfo{number}{3}), pp. \bibinfo{pages}{601--642},
  \doi{10.1007/s11225-021-09972-z}.

\bibitemdeclare{article}{Lukas2022}
\bibitem{Lukas2022}
\bibinfo{author}{Lukas \surnamestart Gr{\"a}tz\surnameend}
  (\bibinfo{year}{2022}): \emph{\bibinfo{title}{Truth tables for modal logics
  {T} and {S4}, by using three-valued non-deterministic level semantics}}.
\newblock {\slshape \bibinfo{journal}{Journal of Logic and Computation}}
  \bibinfo{volume}{32}(\bibinfo{number}{1}), pp. \bibinfo{pages}{129--157},
  \doi{10.1093/logcom/exab068}.

\bibitemdeclare{book}{Humberstone2016}
\bibitem{Humberstone2016}
\bibinfo{author}{Lloyd \surnamestart Humberstone\surnameend}
  (\bibinfo{year}{2016}): \emph{\bibinfo{title}{Philosophical Applications of
  Modal Logic}}.
\newblock \bibinfo{publisher}{College Publications}.

\bibitemdeclare{article}{Ivlev88}
\bibitem{Ivlev88}
\bibinfo{author}{Yuri.~V. \surnamestart Ivlev\surnameend}
  (\bibinfo{year}{1988}): \emph{\bibinfo{title}{{A semantics for modal
  calculi}}}.
\newblock {\slshape \bibinfo{journal}{Bulletin of the Section of Logic}}
  \bibinfo{volume}{17}(\bibinfo{number}{3/4}), pp. \bibinfo{pages}{77--86}.

\bibitemdeclare{book}{Ivlev91}
\bibitem{Ivlev91}
\bibinfo{author}{Yuri.~V. \surnamestart Ivlev\surnameend}
  (\bibinfo{year}{1991}): \emph{\bibinfo{title}{Modal logic. (in Russian)}}.
\newblock \bibinfo{publisher}{Moskva: Moskovskij Gosudarstvennyj Universitet}.

\bibitemdeclare{article}{Kearns81}
\bibitem{Kearns81}
\bibinfo{author}{John \surnamestart Kearns\surnameend} (\bibinfo{year}{1981}):
  \emph{\bibinfo{title}{{Modal Semantics without Possible Worlds}}}.
\newblock {\slshape \bibinfo{journal}{Journal of Symbolic Logic}}
  \bibinfo{volume}{46}(\bibinfo{number}{1}), pp. \bibinfo{pages}{77--86},
  \doi{10.2307/2273259}.

\bibitemdeclare{article}{Kearns89}
\bibitem{Kearns89}
\bibinfo{author}{John \surnamestart Kearns\surnameend} (\bibinfo{year}{1989}):
  \emph{\bibinfo{title}{{Le\'sniewski's strategy and modal logic}}}.
\newblock {\slshape \bibinfo{journal}{Notre Dame Journal of Formal Logic}}
  \bibinfo{volume}{30}(\bibinfo{number}{2}), pp. \bibinfo{pages}{77--86},
  \doi{10.1305/ndjfl/1093635086}.

\bibitemdeclare{inproceedings}{lav:zoh:22}
\bibitem{lav:zoh:22}
\bibinfo{author}{Ori \surnamestart Lahav\surnameend} \& \bibinfo{author}{Yoni
  \surnamestart Zohar\surnameend} (\bibinfo{year}{2022}):
  \emph{\bibinfo{title}{Effective Semantics for the Modal Logics {K} and {KT}
  via Non-deterministic Matrices}}.
\newblock In \bibinfo{editor}{J.~\surnamestart Blanchette\surnameend},
  \bibinfo{editor}{L.~\surnamestart Kov\'acs\surnameend} \&
  \bibinfo{editor}{D.~\surnamestart Pattinson\surnameend}, editors: {\slshape
  \bibinfo{booktitle}{Automated Reasoning. Proceedings of IJCAR 2022}},
  {\slshape \bibinfo{series}{Lecture Notes in Computer Science}}
  \bibinfo{volume}{13385}, \bibinfo{publisher}{Springer, Cham}, pp.
  \bibinfo{pages}{468--485}, \doi{10.1007/978-3-031-10769-6_28}.

\bibitemdeclare{misc}{coniglio2023}
\bibitem{coniglio2023}
\bibinfo{author}{Renato \surnamestart Leme\surnameend},
  \bibinfo{author}{Marcelo \surnamestart Coniglio\surnameend} \&
  \bibinfo{author}{Bruno \surnamestart Lopes\surnameend}
  (\bibinfo{year}{2023}): \emph{\bibinfo{title}{Intuitionism with Truth Tables:
  A Decision Procedure for IPL Based on RNmatrices}}.
\newblock \eprint{2308.13664}.

\bibitemdeclare{article}{Makinson1971}
\bibitem{Makinson1971}
\bibinfo{author}{David \surnamestart Makinson\surnameend}
  (\bibinfo{year}{1971}): \emph{\bibinfo{title}{Some embedding theorems for
  modal logic.}}
\newblock {\slshape \bibinfo{journal}{Notre Dame Journal of Formal Logic}}
  \bibinfo{volume}{12}(\bibinfo{number}{2}), pp. \bibinfo{pages}{252--254},
  \doi{10.1305/ndjfl/1093894226}.

\bibitemdeclare{article}{OmoriSkurt16}
\bibitem{OmoriSkurt16}
\bibinfo{author}{Hitoshi \surnamestart Omori\surnameend} \&
  \bibinfo{author}{Daniel \surnamestart Skurt\surnameend}
  (\bibinfo{year}{2016}): \emph{\bibinfo{title}{More modal semantics without
  possible worlds}}.
\newblock {\slshape \bibinfo{journal}{IfCoLog Journal of Logics and their
  Applications}} \bibinfo{volume}{3}(\bibinfo{number}{5}), pp.
  \bibinfo{pages}{815--846}.

\bibitemdeclare{inproceedings}{OmoriSkurt2020}
\bibitem{OmoriSkurt2020}
\bibinfo{author}{Hitoshi \surnamestart Omori\surnameend} \&
  \bibinfo{author}{Daniel \surnamestart Skurt\surnameend}
  (\bibinfo{year}{2020}): \emph{\bibinfo{title}{A Semantics for a Failed
  Axiomatization of {K}.}}
\newblock In: {\slshape \bibinfo{booktitle}{Advances in Modal Logic}}, pp.
  \bibinfo{pages}{481--501}.

\bibitemdeclare{inproceedings}{OmoriSkurt2021}
\bibitem{OmoriSkurt2021}
\bibinfo{author}{Hitoshi \surnamestart Omori\surnameend} \&
  \bibinfo{author}{Daniel \surnamestart Skurt\surnameend}
  (\bibinfo{year}{2021}): \emph{\bibinfo{title}{Untruth, falsity and
  non-deterministic semantics}}.
\newblock In: {\slshape \bibinfo{booktitle}{2021 IEEE 51st International
  Symposium on Multiple-Valued Logic (ISMVL)}}, \bibinfo{organization}{IEEE},
  pp. \bibinfo{pages}{74--80}, \doi{10.1109/ISMVL51352.2021.00022}.

\bibitemdeclare{inbook}{OmoriSkurt22}
\bibitem{OmoriSkurt22}
\bibinfo{author}{Hitoshi \surnamestart Omori\surnameend} \&
  \bibinfo{author}{Daniel \surnamestart Skurt\surnameend}
  (\bibinfo{year}{2024}): \emph{\bibinfo{title}{On Ivlev's Semantics for
  Modality}}, pp. \bibinfo{pages}{243--275}.
\newblock \bibinfo{publisher}{Springer International Publishing},
  \bibinfo{address}{Cham}, \doi{10.1007/978-3-031-56595-3_9}.

\bibitemdeclare{article}{pawlowski2022}
\bibitem{pawlowski2022}
\bibinfo{author}{Pawel \surnamestart Pawlowski\surnameend} \&
  \bibinfo{author}{Elio \surnamestart La~Rosa\surnameend}
  (\bibinfo{year}{2022}): \emph{\bibinfo{title}{Modular non-deterministic
  semantics for {T, TB, S4, S5} and more}}.
\newblock {\slshape \bibinfo{journal}{Journal of Logic and Computation}}
  \bibinfo{volume}{32}(\bibinfo{number}{1}), pp. \bibinfo{pages}{158--171},
  \doi{10.1093/logcom/exab079}.

\bibitemdeclare{article}{PawlowskiSkurt2023}
\bibitem{PawlowskiSkurt2023}
\bibinfo{author}{Pawel \surnamestart Pawlowski\surnameend} \&
  \bibinfo{author}{Daniel \surnamestart Skurt\surnameend}
  (\bibinfo{year}{2024}): \emph{\bibinfo{title}{8 Valued Non-Deterministic
  Semantics for Modal Logics}}.
\newblock {\slshape \bibinfo{journal}{Journal of Philosophical Logic}},
  \doi{10.1007/s10992-023-09733-4}.

\bibitemdeclare{article}{PawlowskiSkurt2024}
\bibitem{PawlowskiSkurt2024}
\bibinfo{author}{Pawel \surnamestart Pawlowski\surnameend} \&
  \bibinfo{author}{Daniel \surnamestart Skurt\surnameend}
  (\bibinfo{year}{2024}): \emph{\bibinfo{title}{{$\Box$} and {$\Diamond$} in
  eight-valued non-deterministic semantics for modal logics}}.
\newblock {\slshape \bibinfo{journal}{Journal of Logic and Computation}}, p.
  \bibinfo{pages}{exae010}, \doi{10.1007/s10992-023-09733-4}.

\bibitemdeclare{book}{Priest08}
\bibitem{Priest08}
\bibinfo{author}{Graham \surnamestart Priest\surnameend}
  (\bibinfo{year}{2008}): \emph{\bibinfo{title}{Introduction to non-classical
  logics: from ifs to is}}, \bibinfo{edition}{second} edition.
\newblock \bibinfo{publisher}{Cambridge University Press},
  \doi{10.1017/CBO9780511801174}.

\bibitemdeclare{book}{segerberg1971}
\bibitem{segerberg1971}
\bibinfo{author}{Karl~Krister \surnamestart Segerberg\surnameend}
  (\bibinfo{year}{1971}): \emph{\bibinfo{title}{An essay in classical modal
  logic}}.
\newblock \bibinfo{publisher}{Stanford University}.

\bibitemdeclare{article}{wansing1989}
\bibitem{wansing1989}
\bibinfo{author}{Heinrich \surnamestart Wansing\surnameend}
  (\bibinfo{year}{1989}): \emph{\bibinfo{title}{Bemerkungen Zur {S}emantik
  Nicht-Normaler {M}{\"o}glicher {W}elten}}.
\newblock {\slshape \bibinfo{journal}{Mathematical Logic Quarterly}}
  \bibinfo{volume}{35}(\bibinfo{number}{6}), pp. \bibinfo{pages}{551--557},
  \doi{10.1002/malq.19890350611}.

\end{thebibliography}
\end{document}